\documentclass[12pt,english]{article}
\usepackage{libertineRoman}
\usepackage[T1]{fontenc}
\usepackage[latin9]{inputenc}
\usepackage{color}
\usepackage{babel}
\usepackage{verbatim}
\usepackage{amsmath}
\usepackage{amsthm}
\usepackage{amssymb}
\usepackage{setspace}
\usepackage[authoryear]{natbib}
\usepackage{breakurl}

\makeatletter
\theoremstyle{definition}
\newtheorem{defn}{\protect\definitionname}
\theoremstyle{plain}
\newtheorem{assumption}{\protect\assumptionname}
\theoremstyle{plain}
\newtheorem{lem}{\protect\lemmaname}
\theoremstyle{plain}
\newtheorem{prop}{\protect\propositionname}
\theoremstyle{definition}
 \newtheorem{example}{\protect\examplename}
\theoremstyle{plain}
\newtheorem{thm}{\protect\theoremname}
\theoremstyle{plain}
\newtheorem{cor}{\protect\corollaryname}

\usepackage{tikz}
\usetikzlibrary{calc}
\usetikzlibrary{arrows} 
\usetikzlibrary{patterns}

\usepackage{color}
\definecolor{darkblue}{rgb}{0.0,0,.6}
\definecolor{maroon}{rgb}{0.68,0,0}
\definecolor{darkgreen}{rgb}{0,0.369,0.086}
\definecolor{gray}{rgb}{.5,.5,.5}

\usepackage{titletoc}
\usepackage[page,header]{appendix}

\makeatother

\providecommand{\assumptionname}{Assumption}
\providecommand{\corollaryname}{Corollary}
\providecommand{\definitionname}{Definition}
\providecommand{\examplename}{Example}
\providecommand{\lemmaname}{Lemma}
\providecommand{\propositionname}{Proposition}
\providecommand{\theoremname}{Theorem}

\begin{document}

\title{Indicator Choice in Pay-for-Performance}

\author{Majid Mahzoon, Ali Shourideh, Ariel Zetlin-Jones\\
Carnegie Mellon University\thanks{Emails: \textcolor{blue}{\protect\href{mailto:mailto: mmahzoon@andrew.cmu.edu}{mmahzoon@andrew.cmu.edu}},
\textcolor{blue}{\protect\href{mailto:mailto: ashourid@andrew.cmu.edu}{ashourid@andrew.cmu.edu}},
\textcolor{blue}{\protect\href{mailto:mailto: azj@andrew.cmu.edu}{azj@andrew.cmu.edu}}}}
\maketitle
\begin{abstract}
We study the classic principal-agent model when the signal observed
by the principal is chosen by the agent. We fully characterize the
optimal information structure from an agent\textquoteright s perspective
in a general moral hazard setting with limited liability. Due to endogeneity
of the contract chosen by the principal, the agent's choice of information
is non-trivial. We show that the agent's problem can be mapped into
a geometrical game between the principal and the agent in the space
of likelihood ratios. We use this representation result to show that
coarse contracts are sufficient: The agent can achieve her best with
binary signals. Additionally, we can characterize conditions under
which the agent is able to extract the entire surplus and implement
the first-best efficient allocation. Finally, we show that when effort
and performance are one-dimensional, under a general class of models,
threshold signals are optimal. Our theory can thus provide a rationale
for coarseness of contracts based on the bargaining power of the agent
in negotiations.
\end{abstract}

\section{Introduction}

The use of pay-for-performance contracting is a cornerstone of modern
employment contracts. Executive compensation is often indexed in part
to company performance metrics including growth in the company's stock
price. Employment contracts for top athletes frequently involve performance
bonuses for specific outcomes such as goals scored for soccer players
or the number of touch-downs for players in the NFL. When employment
contracts feature such performance pay incentives, employers and employees
must agree to a set of performance indicators during contract negotiations.\footnote{\citet{bebchuk2004pay} discuss the various issues with the negotiations
process between the CEOs and boards and possible issues arising from
choosing particular performance indicators, i.e., vesting stocks.
In soccer, news outlets often describe the process in which players
and soccer clubs agree on what performance measure to use. See for
example \href{https://www.dailymail.co.uk/sport/sportsnews/article-5718227/Premier-League-stars-crazy-contracts-revealed-Alexis-Sanchez-Pierre-Emerick-Aubameyang.html}{this article in Daily Mail}
which describes several soccer players in the Premier League negotiating
over the relevant performance indicator as a base for pay.} Given this observation, what are the incentives of employees and
employers to negotiate on performance indicators? If the employees
have a role in choosing the performance metrics, what metrics would
they choose? Finally, how does the choice of indicators interact with
the ultimate productive efficiency of the firm?

In this paper, we answer these questions by considering the problem
of indicator design in the textbook moral hazard problem with limited
liability. More specifically, we consider the standard principal-agent
problem of \citet{Holmstrom1979} in which an agent has quasi-linear
preferences and must be paid a non-negative wage. The \emph{performance
technology} is one that maps costly effort, $e$, by the agent into
a distribution of some performance measure $x$. Before the principal
offers a compensation contract, the agent chooses an \emph{indicator},
a possibly random signal $s$ of $x$ where the principal can only
offer a contract that is contingent on the indicator, $s$. Once the
indicator is chosen, the principal and the agent play the textbook
moral hazard game.\footnote{Throughout the paper, we assume that the agent commits to the indicator
$s$ while she cannot commit to the effort level $e$. A justification
for this assumption is that employment contracts are often enforced
by courts and thus hard to break.}

In this environment, one might conjecture that more information would
lead to more efficient outcomes. While this is true under certain
circumstances -- see the informativeness principle of \citet{Holmstrom1979},
and its extension by \citet{Chaigneauetal2019} -- information can
often be detrimental to the agent. To see the intuition for this observation,
suppose that the performance technology is degenerate at $x=e$ and
thus the agent can use her effort as the indicator. Then, revealing
all information by choosing $s=x=e$ would give the principal the
ability to fully capture all the surplus generated by the effort.
On the other hand, choosing a fully uninformative indicator leads
to no surplus for the agent as the principal is unable to incentivize
the agent to contribute effort when signals are fully uninformative.
This points to a trade-off in the problem of indicator design by the
agent.

In this model, we show three main results: first, we provide a geometric
interpretation of the indicator design game in the space of likelihoods;
the probability of a particular performance/signal realization for
an arbitrary effort relative to the effort that the agent would like
to implement. Under this interpretation, the agent's indicator design
problem is equivalent to a geometric game where the agent chooses
a convex set (of likelihood ratios) and the principal chooses a point
within that set. Our geometric interpretation provides a tractable
formulation to understand how the agent's choice of indicator influences
the resulting compensation scheme offered by the principal. Second,
we provide conditions under which the agent is able to choose the
indicator in such a way to implement the first-best efficient effort
and capture all the surplus created by her effort. Finally, we consider
a specific case where performance measure $x$ has a continuous distribution
in real numbers and show that under certain conditions optimal indicator
structure takes the form of monotone or hump-shaped thresholds signals.

The reason that the agent is sometimes able to capture all of the
surplus and implement the efficient outcome can be easily understood
when performance technology is fully informative, i.e., $x=e$. In
this case, suppose the agent chooses an indicator with two realizations:
high and low. By choosing the probability of high and low indicators
appropriately for every $x=e$, the agent is able to force the principal's
hand. Note that if the principal wants to implement a given effort
level $\hat{e}$, he must compensate the agent following a realization
of the high signal. The amount of compensation the principal must
deliver is increasing as the cost of the desired effort level $\hat{e}$
rises relative to the cost of any other effort and is increasing in
the likelihood of observing a high signal from some other effort relative
to that of observing a high signal from effort $\hat{e}$. By choosing
the signal probabilities appropriately, the agent is then able to
make the off-path likelihood of the high signal sufficiently large
so that the principal must give up all the surplus in order to implement
the desired effort level $\hat{e}$.

In general, our geometric interpretation allows us to show that the
agent can implement any desired effort level with a ``coarse'' information
structure that has at most two signal realizations. Additionally,
this interpretation allows us to provide conditions under which even
when the performance technology is stochastic, it is possible for
the agent to capture all the surplus and implement the first-best
outcome. Our sufficient condition amounts to checking whether a particular
point in the likelihood space belongs to convex hull of the likelihood
functions implied by the performance technology.  This type of sufficient
condition is easy to evaluate using existing convex hull algorithms.

Beyond our technical contributions, our results shed light on the
debate on the efficiency of incentive contracts in executive compensation.
\citet{bebchuk2004pay} have argued that the standard agency model
of shareholder maximization is at odds with the data since negotiations
often happen between the CEO and the board whose incentives are not
necessarily aligned with that of the shareholders -- in fact they
claim that CEOs and compensation committees often trade favors at
a cost to shareholders. In our model, and consistent with \citet{bebchuk2004pay}'s
interpretation, we endow the agent with full bargaining power over
the choice of performance pay indicators. While the agent captures
all of the gains from trade under this assumption, her choices also
maximize total surplus. In contrast, standard models of moral hazard
with exogenous performance technology often feature inefficiencies
in the sense that providing incentives to the agent often entails
reductions in total surplus. In this sense, we find that optimal design
of performance pay indicators may help to reduce inefficiencies within
the firm. Further investigation of the data on negotiations and choice
of indicators would be a good test of our theory.

\subsection{Related Literature}

Our paper is related to several strands of the literature on contracting
and information design.

With respect to the moral hazard literature, our key innovation is
to consider the problem of choosing indicators whereupon contracts
are based on showing that this can remove all inefficiencies. In the
classical model, \citet{Innes1990} and \citet{PobleteandSpulber2012}
consider moral hazard with limited liability and a risk-neutral agent,
and show that simple debt contracts are optimal. \citet{Carroll2015}
as well as \citet{WaltonandCarroll2022} consider the moral hazard
problem with limited liability, where the principal has non-Bayesian
uncertainty about the production technology and wishes to maximize
her worst-case payoff. \footnote{It is needless to say that a rather large body of work has considered
models with risk-averse agents and risk-neutral principals.}

While often information design incentives are ignored in moral hazard,
\citet{Holmstrom1979} and later \citet{Chaigneauetal2019} are exceptions.
They investigate the comparative statics of changing the performance
technology on the payoffs; by the informativeness principle, the more
informative the output is about the effort, the lower wage the principal
needs to pay.

Perhaps, the closest paper to ours is that of \citet{Garrettetal2020}.
They consider a model in which the agent can design the performance
technology and cost function -- they refer to this as technology
design -- and show that the agent-optimal design involves only binary
distributions. In contrast, in our model, the performance technology
is fixed and the agent chooses an indicator of this performance, i.e.,
an information structure to garble the principal's observations. Moreover,
our paper has a technical contribution by reformulating the problem
in terms of likelihood ratios.

\citet{GeorgiadisandSzentes2020} study moral hazard with limited
liability and a risk-averse agent where the principal continuously
observes signals about the agent's effort at a constant marginal cost.
They show that the principal-optimal information-acquisition strategy
is a two-threshold policy. \citet{Barronetal.2020} consider an agent
who can costlessly add mean-preserving noise to the output. This is
also the case in our model; however, in our setting, the noisy output
is just used for contracting purposes and does not change the principal's
payoff compared to the original output. Moreover, in our model, first
the agent chooses the information structure, then the principal offers
the contract, and finally the agent chooses her effort, whereas in
\citet{Barronetal.2020}, first the principal offers a contract and
then the agent chooses the effort and the information structure.

Another related strand of the moral hazard literature focuses on career
concerns introduced by \citet{Holmstrom1979} and changes in information
structure observed by both the principal and the agent. In this context
\citet{Holmstrom1999} shows that noisy performance signals are beneficial
for incentives. Similarly \citet{Dewatripontetal1999} show that more
informative signals in the sense of Blackwell do not necessarily increase
incentives. In contrast in our model, indicator design or information
design can be used as a tool to  reshuffle surplus between the agent
and the principal while achieving efficiency.

In the context of team production and moral hazard, \citet{halac2021rank}
show that the principal can benefit from private contract offers by
leveraging rank uncertainty: Each agent is informed only of her own
bonus and a ranking distribution; each agent's bonus makes work dominant
if higher-rank agents work. Interestingly, in our setting, it is the
agent that can use uncertainty about performance for the principal
to improve efficiency.

We also contribute to the growing literature on incentives in Bayesian
Persuasion. Several papers, including \citet{BoleslavskyandKim2018},
\citet{Rosar2017}, \citet{PerezRichetandSkreta2022}, \citet{Ball2019},
\citet{SaeediandShourideh2020}, and \citet{Zapechelnyuk2020} have
considered the effect of incentives in the Bayesian persuasion problem
where a third party designs an information structure, and a ``sender''
determines the distribution of the underlying state by exerting a
costly effort. From a technical perspective, our problem is different
from this class of models. This is partly due to the fact that for
the ex-post incentive to exert effort by the agent (after the choice
of indicator and contract), the distribution of the signals -- or
alternatively in the language of \citet{kamenica2011bayesian}, the
distribution of posteriors -- off the equilibrium path is also relevant.
By casting the problem in the space of likelihood functions -- as
opposed to beliefs as in \citet{kamenica2011bayesian} -- we can
characterize its solution using the geometric game. This technique
can be used in other information design problems in which on- and
off-path beliefs are involved.

The rest of the paper is organized as follows: Section \ref{sec:A-Simple-Example}
describes the key insight about full surplus extraction and first-best
implementation in a simple example. Section \ref{sec:Model} describes
the basic model. Section \ref{sec:A-Geometrical-Game} describes the
geometric interpretation of the indicator choice game between the
principal and the agent. Section \ref{sec:Continuous-Effort-and-1}
describes optimality of threshold signals. Section \ref{sec:Conclusion}
concludes.

\section{A Simple Example\label{sec:A-Simple-Example}}

In this section, we use a basic environment to illustrate the main
mechanisms at work in our model. Consider the basic textbook model
of moral hazard. A principal (he) is hiring an agent (she) to perform
a task whose output $x\in\left\{ 0,1\right\} $ is collected by the
principal. The agent chooses how much effort to put in to perform
the task. She can either choose the low effort $e_{L}$ or a costly
high effort $e_{H}$ whose cost is given by $c>0$. 

Choosing the high effort leads to output $x=1$ with certainty while
choosing the low effort leads to output $x=1$ with probability $p<1$
and $x=0$ with probability $1-p$. We assume that the total surplus
under high effort $1-c$ is higher than that under low effort $p$
and thus it is efficient to implement the high effort. 

In this standard principal-agent model with moral hazard, principal
observes the output but not the effort of the agent. He can compensate
the agent for each output realization but cannot make these payments
negative, i.e., he is subject to limited liability. If the principal
sets wage $w$ when output is high, then the agent's incentive compatibility
constraint is
\[
w-c\geq p\cdot w.
\]
Hence, as long as $w\geq\frac{c}{1-p}$, the principal can implement
the high effort. If the principal is to choose $w=\frac{c}{1-p}$,
then his payoff is $1-\frac{c}{1-p}>0$ while the agent's payoff is
$\frac{c}{1-p}-c=\frac{p}{1-p}c>0$. Moreover, for the principal to
prefer implementing the high effort to low, we must have that $1-\frac{c}{1-p}\geq p$
or $c\leq\left(1-p\right)^{2}$.

Now suppose that the agent can control principal's information about
the output. Specifically, suppose that before the contracting stage,
the agent can design a device that can potentially hide the output
of the project. More specifically, suppose that the agent can choose
an information structure or a Blackwell experiment that probabilistically
maps output $x\in\left\{ 0,1\right\} $ to a signal $S=\left\{ L,H\right\} $
which is observed by the principal. The principal observes the signal
$s=H$ with probability $\pi_{H}$ when $x=1$ is realized and observes
$s=H$ with probability $\pi_{L}$ if $x=0$ is realized, where $\pi_{L}<\pi_{H}$. 

Since the principal can only observe the signal designed by the agent,
he will compensate her only when $s=H$. If this compensation is $w$,
then the agent's incentive compatibility constraint is
\[
\pi_{H}\cdot w-c\geq\left(\pi_{H}p+\left(1-p\right)\pi_{L}\right)\cdot w.
\]
 Hence, the principal is able to implement high effort when
\[
w\geq\frac{c}{\left(\pi_{H}-\pi_{L}\right)\left(1-p\right)}.
\]

When minimizing the wage, the expected cost of compensating the agent
for the principal is $\pi_{H}w=\frac{c}{\left(1-\frac{\pi_{L}}{\pi_{H}}\right)\left(1-p\right)}$.
Therefore, as long as $p\leq1-\frac{c}{\left(1-\frac{\pi_{L}}{\pi_{H}}\right)\left(1-p\right)}$,
the principal finds it profitable to implement the high effort. This
inequality can be rewritten as
\[
\frac{c}{\left(1-p\right)^{2}}\leq1-\frac{\pi_{L}}{\pi_{H}}.
\]
This, in turn, implies that when $\frac{c}{\left(1-p\right)^{2}}\leq1$,
the agent can find a signal structure $\left(\pi_{L},\pi_{H}\right)$
such that the above holds with equality. Under such an information
structure, the payoff of the principal is $p$, what he can achieve
without any costly effort, and the agent captures the rest of the
surplus, $1-p-c$. In other words, giving the agent the ability to
choose an information structure enables her to guarantee the highest
value of the surplus under an efficient level of effort. Intuitively,
the change of the information structure allows the agent to induce
an arbitrary high value of the wage by increasing the likelihood $\pi_{L}/\pi_{H}$,
and capture the entire surplus.

The above example illustrates that agent's freedom to choose the information
structure, based on which she will be paid, can be extremely powerful.
A few natural questions arise: When can the agent capture the efficient
level of surplus? What information structure should be used by the
agent to achieve her desired outcome? In what follows, we provide
a characterization of the optimal signal structure for the agent as
well as her ability to extract surplus from the principal.

\section{Model\label{sec:Model}}

Our general model builds upon the textbook moral hazard problem. Consider
a principal employing an agent to perform a task whose output is represented
by $x\in X$, where $X$ is finite. The agent chooses effort $e\in E=\left\{ e_{1},\cdots,e_{m}\right\} $
to perform the task, where $E$ is finite. The agent's effort choice
induces a probability distribution $f(x|e)$ over the outcome space
$X$, where $\sum_{x}f(x|e)=1,\forall e\in E$. We refer to $f\left(\cdot|\cdot\right)$
as the performance technology. Effort is costly to the agent; the
cost of exerting effort $e$ is given by $c(e)$ for some real-valued
function $c:E\rightarrow\mathbb{R}_{+}$. 

Throughout the analysis, we assume that $e_{1}\in E$ represents the
effort with the lowest cost; for simplicity, let $c(e_{1})=0$. The
principal's payoff from realization of output $x$ is given by $g(x)$
for some real-valued function $g:X\rightarrow\mathbb{R}$. The principal
cannot observe the agent's effort and thus cannot offer a contract
contingent on the agent's effort; he can only offer contracts contingent
on observable outcomes. 

The point of departure from the textbook moral hazard model is that
the agent may influence the principal's information about the output
by choosing an information structure $(S,\pi)$. Here, $S$ is a signal
space and $\pi(\cdot|x):X\rightarrow\Delta(S)$ is a stochastic mapping
from the output space $X$ to the signal space, where $\sum_{s}\pi(s|x)=1,\ \forall x\in X$.
The principal only observes the signal $s\in S$ generated from this
information structure and can thus only offer a contract contingent
on this signal realization. Therefore, the principal's choice of contract
can be represented by a real-valued function $w:S\rightarrow\mathbb{R}_{+}$
where $w(s)$ is the wage paid to the agent when signal $s$ is realized.
One interpretation of this contractual restriction is that the task
output is not observable to the principal but the agent may verifiably
disclose information about the output. An alternative interpretation
is that the output is observable but the agent has the option to choose
the performance measure based on which she will be paid.

Note that as in the simple example, we have assumed that agent enjoys
limited liability, i.e., the contract offered to her by the principal
guarantees a non-negative wage regardless of the effort she puts in.
The principal's payoff $u_{P}$ is equal to the payoff from output
less the wages paid to the agent. The agent's payoff $u_{A}$ is equal
to the wage she receives from the principal minus the cost of her
effort. Both the principal and the agent are assumed to maximize expected
utility. Notice that principal can always implement the zero-cost
effort, i.e. $e_{1}$, by offering $w(s)=0,\ \forall s\in S$. Therefore,
his outside option is to implement $e_{1}$ and obtain $\underline{u}_{P}=\sum_{x}g(x)f(x|e_{1})$. 

The timing of the game is as follows: 
\begin{itemize}
\item Agent chooses an information structure $(S,\pi)$. To ease the exposition,
we assume the signal space $S$ is finite and simply represent the
information structure by $\pi$. \footnote{Later, we show this restriction to finite signal spaces is without
loss of generality.}
\item Observing the information structure $(S,\pi)$ chosen by the agent,
the principal offers the agent a contract $w:S\rightarrow\mathbb{R}_{+}$
contingent on the realized signal. 
\item Observing the contract $w$ offered by the principal (and the information
structure $(S,\pi)$ she has chosen), the agent chooses how much effort
$e$ to exert. 
\item Given the effort $e$ chosen by the agent, output $x$ is realized
according to $f(x|e)$ and then signal $s\in S$ is realized according
to $\pi(s|x)$. 
\item Payoffs are realized where agent's payoff is $u_{A}=w(s)-c(e)$, and
the principal's payoff is $u_{P}=g(x)-w(s)$. 
\end{itemize}
To summarize, the game is played sequentially in three stages. In
the first stage, the agent chooses the information structure $\pi$
to (partially) inform the principal about the realized output. In
the second stage, the principal offers the agent a contract $w$ contingent
on the signal realization. In the last stage, the agent chooses her
effort $e$.

The equilibrium concept is the standard subgame perfect equilibrium
(SPE). The agent's strategy $(\pi,\sigma_{e}(\pi,w))$ consists of
a choice of information structure $\pi$ and a choice of effort $\sigma_{e}(\pi,w)$
as a function of $\pi$ and the contract $w$ offered by the principal.
The principal's strategy $\sigma_{w}(\pi)$ is the contract he offers
to the agent as a function of the information structure $\pi$ chosen
by the agent.
\begin{defn}
An SPE of the game $\sigma^{*}=(\pi^{*},\sigma_{e}^{*}(\pi,w),\sigma_{w}^{*}(\pi))$
consists of a strategy for the agent $(\pi^{*},\sigma_{e}^{*}(\pi,w))$
and a strategy for the principal $\sigma_{w}^{*}(\pi)$ such that: 
\begin{itemize}
\item $\sigma_{e}^{*}(\pi,w)$ maximizes the agent's expected utility for
every information structure $\pi$ and every principal's choice of
contract $w$. 
\item $\sigma_{w}^{*}(\pi)$ maximizes the principal's expected utility
for every agent's choice of information structure $\pi$ given agent's
equilibrium effort strategy $\sigma_{e}^{*}(\pi,w)$. 
\item $\pi^{*}$ maximizes the agent's expected utility given principal's
equilibrium strategy $\sigma_{w}^{*}(\pi)$ and her own equilibrium
effort strategy $\sigma_{e}^{*}(\pi,w)$. 
\end{itemize}
\end{defn}
We let $p(\cdot|e):E\rightarrow\Delta(S)$ represent the resulting
stochastic mapping from the effort space $E$ to the signal space
$S$ induced by a given information structure $(S,\pi)$ and the underlying
probability distribution of outcomes given effort. Note, for any level
of effort $e$ and any signal realization $s$, $p(s|e)=\sum_{x}f(x|e)\pi(s|x)$
and $\sum_{s}p(s|e)=1,\ \forall e\in E$. Conditional on a given information
structure $(S,\pi)$, both the principal and the agent use the stochastic
mapping $p\left(\cdot|\cdot\right)$ to evaluate their expected payoffs.

We now formulate the problems the principal and the agent solve beginning
with the agent's choice of effort. The agent's expected utility is
$U_{A}=\sum_{s}p(s|e)w(s)-c(e)$. Therefore, the agent's problem in
the last stage of the game (where $\pi$ and $w$ have previously
been chosen in the preceding stages) is 

\begin{equation}
\max_{e}\ \sum_{s}p(s|e)w(s)-c(e).\tag{IC-A}\label{ICA}
\end{equation}

The principal's expected utility is $U_{P}=\sum_{x}g(x)f(x|e)-\sum_{s}p(s|e)w(s)$
if he chooses to implement effort $e$. It is convenient to define
$\mathbb{E}[g(x)|e]=\sum_{x}g(x)f(x|e)$. For a given desired effort
level $e$, the principal chooses a wage schedule that solves 
\begin{align}
\begin{split}\min_{w}\  & \sum_{s}p(s|e)w(s)\ \text{s.t.}\\
 & \sum_{s}p(s|e)w(s)-c(e)\geq\sum_{s}p(s|\hat{e})w(s)-c(\hat{e}),\ \forall\hat{e}\in E,\\
 & w(s)\geq0,\ \forall s\in S.
\end{split}
\label{eq: principal}
\end{align}
The constraints represent the agent's incentive compatibility and
a set of limited liability constraints. Notice that we have not imposed
a participation constraint for the agent. This is because the assumption
$c(e_{1})=0$ together with limited liability implies that setting
$e=e_{1}$ guarantees a non-negative payoff for the agent. Let $W(e,\pi)$
represent the optimal value in (\ref{eq: principal}). This is the
minimum expected wage the principal must pay the agent to implement
effort $e$ given the information structure $\pi$ chosen by the agent.

Given $W(e,\pi)$, the problem of the principal is to choose an effort
level to maximize her expected utility, or, 
\begin{equation}
\max_{e}\ \mathbb{E}\left[g\left(x\right)|e\right]-W(e,\pi).\label{principal_simplified}
\end{equation}

It is possible to express the above as an incentive compatibility
constraint and thus write the agent's problem in the first stage of
the game as 
\[
\max_{e,\pi}\ W(e,\pi)-c(e)
\]
subject to the principal's incentive compatibility constraint
\begin{equation}
\mathbb{E}\left[g\left(x\right)|e\right]-W(e,\pi)\geq\mathbb{E}\left[g\left(x\right)|\hat{e}\right]-W(\hat{e},\pi),\ \forall\hat{e}\in E.\tag{IC-P}\label{eq:ICP}
\end{equation}

\subsection{Remarks on the Environment}

It is useful to discuss various interpretations of the model as well
as our key assumptions.

\subsubsection*{Performance Measure as Information Structure}

In the textbook model of moral hazard, the principal cannot observe
the agent's effort. He therefore uses some imperfect signal of effort
to incentivize the agent. If he can observe the output, he offers
an output-contingent contract; this makes the output the performance
measure for the agent. In our model, the principal does not observe
the output but does observe a signal, which may be correlated with
effort, and he offers a signal-contingent contract. As a result, the
signal is the relevant performance measure for the agent. By choosing
the information structure, the agent influences the set of observables
that will ultimately dictate her compensation, which we interpret
as the agent choosing her performance measure.

We make the assumption that the principal cannot offer output-contingent
contracts. As we have described, while a conventional interpretation
of this assumption is that the principal cannot observe the output
directly, an alternative interpretation is that this restriction arises
during the negotiations between the agent and the principal in choosing
contractual performance measures. 

\subsubsection*{Commitment}

We assume that the agent commits to an information structure in the
first stage of the game. Our interpretation of the information structure
as a contractual performance measure provides a natural justification
of the commitment assumption. When signing a contract, all parties
are aware of and agree on the probabilistic nature of the chosen performance
measure as a function of the output. The contractible nature of the
performance measure makes the commitment assumption necessary.

\subsubsection*{Comparison to the Literature on Bayesian Persuasion}

As in the Bayesian persuasion literature, we can write the problem
in terms of the distribution of posteriors induced by the information
structure. In the Bayesian persuasion, every choice of information
structure induces a distribution of posteriors, where the whole distribution
matters: not only the support, but also the probabilities. In our
setting, every choice of information structure induces a set of distributions
of posteriors, one for each choice of effort. These distributions
are related through their supports: given the support of one distribution,
the supports of the others are pinned down. As we will discuss in
section \ref{sec:A-Geometrical-Game}, the key sufficient statistic
about the choice of information structure  is the distribution of
the likelihood ratios and these are determined by the supports of
the above distributions. Therefore, in our model, only the support
of the distribution of posteriors matters. 

\section{A Geometric Analysis of the Game\label{sec:A-Geometrical-Game}}

We now characterize the equilibrium outcomes of the game. To do so,
we first describe the set of effort levels that are implementable
by some information structure. We then show a ``coarse''-ness result.
That is, we show that it is without loss of generality to restrict
the agent's choice of information structures to binary structures,
where the set of signals has only two discrete points. Using such
information structures, in the next section, we derive sufficient
conditions such that the first-best level of effort is implementable.
When these sufficient conditions are satisfied, we show that the agent
chooses an information structure that implements the first-best effort
level and extracts the entire surplus.

While it is possible to work with zero probability events and define
likelihood ratios -- by describing how division by 0 is defined --
in order to avoid complications, we make the following assumption:
\begin{assumption}
The performance technology is full support, i.e., $\forall x\in X,\forall e\in E,f\left(x|e\right)>0$.
\end{assumption}
This assumption implies that all the likelihood ratios below are well-defined.

\textbf{Implementable Effort.} To characterize the set of implementable
effort levels, we use backward induction and first re-cast the problem
of the principal geometrically. Specifically, we show that the likelihood
ratios for each signal realization $s$ for any effort level $e$
relative to the desired implementable level $e^{*}$ are sufficient
statistics to solve the principal's problem. In other words, we argue
that any desired effort level to implement $e^{*}$ and any information
structure $(S,\pi)$ give rise to a (geometric) space of possible
likelihood ratios and that the principal's optimal choice of compensation
schemes may be reduced to choosing a point in this space of likelihood
ratios.

More formally, we define an implementable effort level $e^{*}$ as
follows:
\begin{defn}
An effort level $e^{*}$ is \emph{implementable} if there exists an
information structure $(S,\pi)$ such that $e^{*}$ is a solution
to the principal's problem (\ref{principal_simplified}).
\end{defn}
Given this definition, an implementable effort level $e^{*}$ must
satisfy the two incentive compatibility constraints in (\ref{eq:ICP})
and (\ref{ICA}) for the principal and the agent, where the agent's
incentive compatibility constraint must hold for all possible histories
including those following a deviation by the principal that involves
recommending an alternative level of effort.

Let $e^{*}$ represent some effort level the agent would like to implement
(in the first stage of the game). To motivate the relevance of likelihood
ratios, consider the agent's interim incentive compatibility constraint
when the principal only pays the agent following a signal realization
$s$:
\[
p\left(s|e^{*}\right)w\left(s\right)-c\left(e^{*}\right)\geq p\left(s|\hat{e}\right)w\left(s\right)-c\left(\hat{e}\right),\forall\hat{e}\in E.
\]
These constraints imply that for any effort level $\hat{e}$ where
signal $s$ is less likely than under $e^{*}$, i.e., $p\left(s|\hat{e}\right)<p\left(s|e^{*}\right)$,
the wage must satisfy
\begin{equation}
w\left(s\right)\geq\frac{c\left(e^{*}\right)-c\left(\hat{e}\right)}{p\left(s|e^{*}\right)-p\left(s|\hat{e}\right)},\label{eq: IC2}
\end{equation}
and if $s$ is more likely under $\hat{e}$ than under $e^{*}$ then
\begin{equation}
\frac{c\left(\hat{e}\right)-c\left(e^{*}\right)}{p\left(s|\hat{e}\right)-p\left(s|e^{*}\right)}\geq w\left(s\right).\label{eq: IC3}
\end{equation}
The first set of constraints (\ref{eq: IC2}) imply that the expected
wage must satisfy
\[
p\left(s|e^{*}\right)w\left(s\right)\geq\max_{\hat{e}:p\left(s|e^{*}\right)>p\left(s|\hat{e}\right)}\frac{c\left(e^{*}\right)-c\left(\hat{e}\right)}{1-\frac{p\left(s|\hat{e}\right)}{p\left(s|e^{*}\right)}}.
\]
In other words, the expected cost of implementing $e^{*}$ for the
principal (and its implicit benefit for the agent) is determined by
the likelihood of signal $s$. The second set of constraints (\ref{eq: IC3})
place an upper bound on the wages the principal may deliver while
respecting incentives of the agent. As we show below, this restriction
can also be formulated in terms of likelihood ratios. 

The above illustration only holds under the assumption that the principal
only compensates the agent following a single signal $s$. We now
show a more general version of this analysis for arbitrary compensation
schemes. To this end, consider an arbitrary wage schedule $w\left(s\right)$
chosen by the principal -- for any effort $e_{i}\in E$ chosen by
the principal on- or off- equilibrium path . We may write the expected
wage paid to the worker as
\begin{align*}
\sum_{s}w\left(s\right)p\left(s|e_{i}\right) & =\sum_{s}w\left(s\right)p\left(s|e^{*}\right)\frac{p\left(s|e_{i}\right)}{p\left(s|e^{*}\right)}\\
 & =\sum_{s}w\left(s\right)p\left(s|e^{*}\right)\cdot\sum_{s}\frac{w\left(s\right)p\left(s|e^{*}\right)}{\sum_{s}w\left(s\right)p\left(s|e^{*}\right)}\frac{p\left(s|e_{i}\right)}{p\left(s|e^{*}\right)}\\
 & =\sum_{s}w\left(s\right)p\left(s|e^{*}\right)\cdot\sum_{s}\alpha_{s}\frac{p\left(s|e_{i}\right)}{p\left(s|e^{*}\right)},
\end{align*}
where $\sum_{s}\alpha_{s}=1$. Since the weights $\alpha_{s}$ do
not depend on $e_{i}$, we may write the agent's interim incentive
compatibility constraint as
\[
\sum_{s}w\left(s\right)p\left(s|e^{*}\right)\cdot\sum_{s}\alpha_{s}\frac{p\left(s|e_{i}\right)}{p\left(s|e^{*}\right)}-c\left(e_{i}\right)\geq\sum_{s}w\left(s\right)p\left(s|e^{*}\right)\cdot\sum_{s}\alpha_{s}\frac{p\left(s|e_{j}\right)}{p\left(s|e^{*}\right)}-c\left(e_{j}\right).
\]
Therefore, if we define $\ell_{i}=1-\sum_{s}\alpha_{s}p\left(s|e_{i}\right)/p\left(s|e^{*}\right)$,
this incentive constraint may be written as
\begin{equation}
\sum_{s}w\left(s\right)p\left(s|e^{*}\right)\cdot\left[\ell_{j}-\ell_{i}\right]\geq c\left(e_{i}\right)-c\left(e_{j}\right),\forall j=1,\cdots,\left|E\right|.\label{eq:ICAr}
\end{equation}
Writing the incentive constraint in this manner reveals that the choice
of likelihood ratios, $\ell_{j}$'s, is sufficient to characterize
payoffs of the principal and the agent and that the choice of contract
$w\left(s\right)$ by the principal may be decomposed into a choice
of $\sum_{s}w\left(s\right)p\left(s|e^{*}\right)$ as well as the
choice of $\left\{ \alpha_{s}\right\} $, or alternatively, $\ell_{j}$'s.
In other words, if we represent the likelihood ratio profile with
the following vector
\[
\mathbf{l}=\left(\ell_{1},\cdots,\ell_{\left|E\right|}\right),
\]
then $\mathbf{l}\in\text{convex hull}\left(\left\{ \big(1-\frac{p\left(s|e_{1}\right)}{p\left(s|e^{*}\right)},\cdots,1-\frac{p\left(s|e_{\left|E\right|}\right)}{p\left(s|e^{*}\right)}\big)\right\} _{s\in S}\right)=\text{co}\left(\mathbf{p}\right)$.
Thus the choice of the principal can be summarized by the overall
level of the compensation $\sum_{s}p\left(s|e^{*}\right)w\left(s\right)=\overline{w}$
together with an element of the set $\text{co}\left(\mathbf{p}\right)$.
We thus have the following lemma:
\begin{lem}
Consider an implementable effort $e^{*}$ and its associated information
structure $\mathbf{p}=\left\{ p\left(\cdot|\cdot\right)\right\} $.
Then the cost to the principal from implementing any effort $e_{i}$
is given by 
\begin{equation}
W(e_{i},\mathbf{p})=\min_{\mathbf{l}\in\text{co}(\mathbf{p})\cap\Omega_{i}}(1-\ell_{i})\cdot\max_{j:\ell_{j}\geq\ell_{i}}\frac{c(e_{i})-c(e_{j})}{\ell_{j}-\ell_{i}},\label{wage_space}
\end{equation}
where 
\begin{equation}
\Omega_{i}=\bigg\{\mathbf{l}\in\mathbb{R}^{\left|E\right|}:\max_{j:\ell_{j}\geq\ell_{i}}\frac{c\left(e_{i}\right)-c\left(e_{j}\right)}{\ell_{j}-\ell_{i}}\leq\min_{j:\ell_{j}<\ell_{i}}\frac{c\left(e_{i}\right)-c\left(e_{j}\right)}{\ell_{j}-\ell_{i}}\bigg\}.\label{feasible_space}
\end{equation}
\end{lem}
The above lemma states that the problem of the principal can be reduced
to choosing an effort level as well as a point in the convex hull
of likelihood ratios -- its intersection with the convex cone in
(\ref{feasible_space}). Note that not all likelihood ratio profiles
$\mathbf{l}$ are feasible in the sense that there may be no compensation
scheme that satisfies the agent's interim incentive compatibility
constraint. As in our motivating example above, the likelihood ratios
must permit a wage $w(s)$ that lies between the lower and upper bounds
in (\ref{feasible_space}). The convex cone $\Omega_{i}$ defines
the set of likelihood ratio profiles that admit incentive compatible
compensation schemes.

These results reveal that by choosing an information structure, the
agent effectively determines the convex hull $\text{co}\left(\mathbf{p}\right)$
and then the principal chooses a point that is in the intersection
of $\text{co}\left(\mathbf{p}\right)$ and the convex cone $\Omega_{i}$.
This result by itself does not make the analysis more tractable as
it does not immediately describe the set of feasible convex hulls
$\text{co\ensuremath{\left(\mathbf{p}\right)}}$. The following proposition
provides such a characterization. Note that, similar to $\text{co}\left(\mathbf{p}\right)$,
the convex hull $\text{co}\left(\mathbf{f}\right)$ is the convex
hull created by the points $\left\{ \bigg(1-\frac{f\left(x|e_{1}\right)}{f\left(x|e^{*}\right)},\cdots,1-\frac{f\left(x|e_{\left|E\right|}\right)}{f\left(x|e^{*}\right)}\bigg)\right\} _{x\in X}$.
\begin{prop}
\label{prop:For-any-information}For any information structure $\left(S,\pi\right)$
with $\left|S\right|<\infty$, its associated $\text{co}\left(\mathbf{p}\right)$
is a subset of $\text{co}\left(\mathbf{f}\right)$ that contains the
origin $\mathbf{0}=\left(0,\cdots,0\right)$. Additionally, for any
convex subset $C$ of $\text{co}\left(\mathbf{f}\right)$ that contains
the origin and has a finite set of extreme points, there exists an
information structure $\left(S,\pi\right)$ such that $\text{co}\left(\mathbf{p}\right)=C$.
\end{prop}
\begin{proof}
Let $\left(S,\pi\right)$ be an information structure. Then,
\begin{align*}
\frac{p(s|e_{i})}{p(s|e^{*})} & =\frac{\sum_{x}f(x|e_{i})\pi(s|x)}{\sum_{x}f(x|e^{*})\pi(s|x)}=\sum_{x}\frac{f(x|e^{*})\pi(s|x)}{\sum_{x}f(x|e^{*})\pi(s|x)}\frac{f(x|e_{i})}{f(x|e^{*})}\\
 & =\sum_{x}\beta_{s}\left(x\right)\frac{f(x|e_{i})}{f(x|e^{*})},\forall i=1,\cdots,\left|E\right|,
\end{align*}
where $\sum_{x}\beta_{s}\left(x\right)=1$; $\beta_{s}\left(x\right)$
is the principal's posterior probability of $x$ after observing $s$.
The above implies that
\[
\left(1-\frac{p(s|e_{1})}{p(s|e^{*})},\cdots,1-\frac{p(s|e_{\left|E\right|})}{p(s|e^{*})}\right)=\sum_{x}\beta_{s}\left(x\right)\left(1-\frac{f(x|e_{1})}{f(x|e^{*})},\cdots,1-\frac{f(x|e_{\left|E\right|})}{f(x|e^{*})}\right).
\]
Hence, the left hand side of the above is a member of $\text{co}\left(\mathbf{f}\right)$
for all $s\in S$. As a result $\text{co}\left(\mathbf{p}\right)\subset\text{co}\left(\mathbf{f}\right)$.
Moreover, we have 
\[
\sum_{s}p(s|e^{*})\left(1-\frac{p(s|e_{i})}{p(s|e^{*})}\right)=1-\sum_{s}p(s|e_{i})=0,\forall i=1,\cdots,\left|E\right|.
\]
This implies that the convex set $\text{co}(\mathbf{p})$ includes
the origin.

Now consider an arbitrary convex set $C\subset\text{co}\left(\mathbf{f}\right)$
that contains the origin with each of its members being of the form
$\mathbf{z}=\left(z_{1},\cdots,z_{\left|E\right|}\right)$. Let $S$
be the set of extreme points of $C$. Then, since $\mathbf{0}\in C$,
by Caratheodory theorem -- see \citet{rockafellar1970convex} --
there must exist $\left\{ \tau_{\mathbf{z}}\right\} _{\mathbf{z}\in S}$
such that $\sum_{\mathbf{z\in}S}\tau_{\mathbf{z}}=1$ and
\begin{equation}
0=\sum_{\mathbf{z}\in S}\tau_{\mathbf{z}}z_{i},\forall i=1,\cdots,\left|E\right|.\label{eq: tauz}
\end{equation}
By definition of $\text{co}\left(\mathbf{f}\right)$, we must have
$z_{i}\leq1$ for all $\mathbf{z}\in C$. Moreover, since $\mathbf{z\in\text{co}\left(\mathbf{f}\right)}$,
there must exist a subset $Y\subset X$ whose members are linearly
independent together with $\beta_{\mathbf{z}}\left(x\right)$ such
that
\[
\sum_{x\in Y}\beta_{\mathbf{z}}\left(x\right)=1,\beta_{\mathbf{z}}\left(x\right)\geq0,z_{i}=\sum_{x\in Y}\beta_{\mathbf{z}}\left(x\right)\left[1-\frac{f\left(x|e_{i}\right)}{f\left(x|e^{*}\right)}\right],\forall i=1,\cdots,\left|E\right|.
\]
Replacing the above in (\ref{eq: tauz}) leads to
\[
0=\sum_{x\in Y}\sum_{\mathbf{z}\in S}\tau_{\mathbf{z}}\beta_{\mathbf{z}}\left(x\right)\left[1-\frac{f\left(x|e_{i}\right)}{f\left(x|e^{*}\right)}\right],\forall i=1,\cdots,\left|E\right|.
\]
Since the points in $Y$ are linearly independent and we also know
that
\[
0=\sum_{x\in Y}f\left(x|e^{*}\right)\left[1-\frac{f\left(x|e_{i}\right)}{f\left(x|e^{*}\right)}\right],\forall i=1,\cdots,\left|E\right|,
\]
we must have that 
\[
f\left(x|e^{*}\right)=\sum_{\mathbf{z}}\tau_{\mathbf{z}}\beta_{\mathbf{z}}\left(x\right).
\]
Let us define
\[
\pi\left(\mathbf{z}|x\right)=\frac{\beta_{\mathbf{z}}\left(x\right)\tau_{\mathbf{z}}}{\sum_{\mathbf{z}\in S}\beta_{\mathbf{z}}\left(x\right)\tau_{\mathbf{z}}}=\frac{\beta_{\mathbf{z}}\left(x\right)\tau_{\mathbf{z}}}{f\left(x|e^{*}\right)}.
\]
Then under $\pi\left(\mathbf{z}|x\right)$, we have
\begin{align*}
1-\frac{p\left(\mathbf{z}|e_{i}\right)}{p\left(\mathbf{z}|e^{*}\right)} & =1-\frac{\sum_{x}\frac{\beta_{\mathbf{z}}\left(x\right)\tau_{\mathbf{z}}}{f\left(x|e^{*}\right)}f\left(x|e_{i}\right)}{\sum_{x}\frac{\beta_{\mathbf{z}}\left(x\right)\tau_{\mathbf{z}}}{f\left(x|e^{*}\right)}f\left(x|e^{*}\right)}\\
 & =1-\frac{\sum_{x}\beta_{\mathbf{z}}\left(x\right)\tau_{\mathbf{z}}\frac{f\left(x|e_{i}\right)}{f\left(x|e^{*}\right)}}{\tau_{\mathbf{z}}\sum_{x}\beta_{\mathbf{z}}\left(x\right)}\\
 & =1-\sum_{x}\beta_{\mathbf{z}}\left(x\right)\frac{f\left(x|e_{i}\right)}{f\left(x|e^{*}\right)}=z_{i},\forall i=1,\cdots,\left|E\right|,
\end{align*}
 which concludes the proof.
\end{proof}
Proposition \ref{prop:For-any-information} implies that any convex
subset of $\text{co}\left(\mathbf{f}\right)$ that contains the origin
can be chosen by the agent as $\text{co}\left(\mathbf{p}\right)$.
This implies that instead of a choice of $\mathbf{p}$, we can focus
on a choice of a convex subset of $\text{co}\left(\mathbf{f}\right)$
that contains the origin. Thus the agent-optimal information structure
that implements $e^{*}$ can be thought of as the equilibrium of the
following game:
\begin{itemize}
\item \textbf{Stage 1. }The agent chooses a finite set of points $L$ inside
the convex set $\text{co}\left(\mathbf{f}\right)$ such that the convex
hull of these points $\text{conv}\left(L\right)$ includes the origin.
\item \textbf{Stage 2. }The principal chooses an effort level $e_{i}\in E$
and a point $\ell\in\text{conv}\left(L\right)\cap\Omega_{i}$ to maximize
$\mathbb{E}\left[g\left(x\right)|e_{i}\right]-\left(1-\ell_{i}\right)\cdot\max_{j:\ell_{j}>\ell_{i}}\frac{c(e_{i})-c(e_{j})}{\ell_{j}-\ell_{i}}$.
\item \textbf{Stage 3. }The choice of principal in stage 2 coincides with
$e^{*}$.
\end{itemize}
The following example illustrates the construction of $\text{co}\left(\mathbf{f}\right)$
and the equilibrium response of the principal.
\begin{example}
\label{exa:Suppose-that-,}Suppose that $X=\left\{ x_{1}=0,x_{2}=1,x_{3}=2\right\} $
and $E=\left\{ e_{1},e_{2},e_{3}\right\} $, where $c\left(e_{1}\right)=0$,
$c\left(e_{2}\right)=0.1$, and $c\left(e_{3}\right)=0.3$. The performance
technology is given by 
\[
f=\begin{bmatrix}0.35 & 0.50 & 0.15\\
0.05 & 0.50 & 0.45\\
0.10 & 0.15 & 0.75
\end{bmatrix},
\]
where $f_{ij}=f\left(x_{j}|e_{i}\right)$. Moreover, the principal\textquoteright s
payoff from realization of output $x$ is equal to $x$, that is,
$g(x)=x$. This yields 
\[
\mathbb{E}\left[g(x)|e_{1}\right]=0.8,\ \mathbb{E}\left[g(x)|e_{2}\right]=1.4,\ \mathbb{E}\left[g(x)|e_{3}\right]=1.65.
\]
Therefore, the first-best is to implement effort $e_{3}$. If we set
$e^{*}=e_{3}$, then since $1-f\left(x|e_{3}\right)/f\left(x|e^{*}\right)=0$,
we can embed the set $\text{co}\left(\mathbf{f}\right)$ in $\mathbb{R}^{2}$.
The gray area in Figure \ref{fig:The-geometric-representation} represents
this set with $a_{i}$ being the point associated with $x_{i}\in X$. 

To understand the geometry of the game between the principal and the
agent, consider a signal structure with 4 realizations given by $S=\left\{ s_{1},s_{2},s_{3},s_{4}\right\} $
and
\[
\pi=\begin{bmatrix}0.5 & 0.2 & 0.1 & 0.2\\
0.2 & 0.3 & 0.3 & 0.2\\
0.05 & 0.05 & 0.1 & 0.8
\end{bmatrix},
\]
where $\pi_{ij}=\pi(s_{j}|x_{i})$. We can use $p\left(s|e\right)=\sum_{x\in X}\pi\left(s|x\right)f\left(x|e\right)$
to construct the likelihood ratios. The green area in Figure \ref{fig:The-geometric-representation}
illustrates the convex set $\text{co}(\mathbf{p})$.

To better illustrate the incentives in the geometric game between
the principal and the agent, suppose the agent reveals $x$ to the
principal. As we have illustrated above, the choice of contract by
the principal is equivalent to the choice of a point in the convex
set $\text{co}\left(\mathbf{f}\right)$.

To think about the incentives of the principal, if the principal is
to implement $e_{3}$, he chooses $\mathbf{l}\in\text{co}\left(\mathbf{f}\right)$
to minimize its cost given by 
\[
\max\left\{ \frac{c\left(e_{3}\right)-c\left(e_{2}\right)}{\ell_{2}},\frac{c\left(e_{3}\right)-c\left(e_{1}\right)}{\ell_{1}}\right\} =\max\left\{ \frac{0.2}{\ell_{2}},\frac{0.3}{\ell_{1}}\right\} 
\]
The lower contour sets associated with the above cost function are
convex cones in the shape of positive orthants -- the shaded blue
area highlighted in the right panel of Figure \ref{fig:The-geometric-representation}.
For this example, the above cost function is minimized at $\mathbf{\mathbf{l^{*}}}$.
On the other hand, if the principal is to implement $e_{2}$, he chooses
$\mathbf{l}\in\text{co}\left(\mathbf{f}\right)$ to minimize the cost
given by $\left(1-\ell_{2}\right)\frac{0.1}{\ell_{1}-\ell_{2}}$ when
$\ell_{1}\geq\ell_{2}$. This cost is minimized at $a_{3}$. For $e_{3}$
to be implementable, the principal's payoff associated with $e_{2}$
-- implied by the principal choosing $a_{3}$ -- should be lower
than his payoff associated with $e_{3}$-- implied by the principal
choosing$\mathbf{\mathbf{l^{*}}}$. The set of likelihood ratios that
satisfy this condition are highlighted by the blue shaded area in
the right panel of Figure \ref{fig:The-geometric-representation}.
Note that $\mathbf{\mathbf{l^{*}}}$ is not contained in this area.
As a result, by choosing to reveal $x$ to the principal, the agent
is unable to implement $e_{3}$. One can show that $e_{2}$ is implementable
under full revelation of $x$.

\begin{figure}[t]
\begin{centering}
\begin{tikzpicture}[scale=1.3]
\clip (-2.9,-2.9) rectangle (2.7,2.7);
\fill[fill=gray!20!white] (0.8,0.4) -- (-7/3,-7/3) -- (-2.5,0.5) -- cycle;
\filldraw[fill=green!20!white, draw=green!50!black,thick] (-1.4043,-0.2553) -- (-1.2195,-0.7805) -- (-0.5385,-0.5385) -- (0.5538, 0.2769) -- cycle;
\fill[fill=gray] (0.8,.4) circle [radius=0.05] node[anchor=west]{$a_3$};
\fill[fill=gray] (-7/3,-7/3) circle [radius=0.05] node[anchor=north]{$a_2$};
\fill[fill=gray] (-2.5,.5) circle [radius=0.05] node[anchor=east]{$a_1$};
\fill[fill=green!40!black] (-1.4043,-0.2553) circle [radius=0.05];
\fill[fill=green!40!black] (-1.2195,-0.7805) circle [radius=0.05];
\fill[fill=green!40!black] (-0.5385,-0.5385) circle [radius=0.05];
\fill[fill=green!40!black] (0.5538, 0.2769) circle [radius=0.05];
\draw[->, very thick] (-10,0) -- (1.5,0);
\draw[->, very thick] (0,-10) -- (0,1.5);
\filldraw (0,0) circle [radius=0.05];
\draw[step=1,gray,very thin] (-5,-5) grid (1.5,1.5);
\draw (0,1.7) node {$\ell_2$};
\draw (1.7,0) node {$\ell_1$};
\end{tikzpicture}\begin{tikzpicture}[scale=1.3]
\clip (-2.9,-2.9) rectangle (2.7,2.7);   
\fill[fill=gray!20!white] (0.8,0.4) -- (-7/3,-7/3) -- (-2.5,0.5) -- cycle;   
\draw[->, very thick] (-10,0) -- (1.5,0);   
\draw[->, very thick] (0,-10) -- (0,1.5);   
\filldraw (0,0) circle [radius=0.05];   
\fill[fill=gray] (0.8,.4) circle [radius=0.05] node[below]{$a_3$};
\fill[fill=gray] (-7/3,-7/3) circle [radius=0.05] node[anchor=north]{$a_2$};
\fill[fill=gray] (-2.5,.5) circle [radius=0.05] node[anchor=north]{$a_1$};
\fill[fill=gray] (0.6087,0.4058) circle [radius=0.05] node[above]{$\mathbf{l}^*$};
\draw[dashed, very thick, draw=blue] (3/4,1/2) -- (3/4,1.5);
\draw[dashed, very thick, draw=blue] (3/4,1/2) -- (1.5,1/2);
\filldraw[fill=blue!30!white, draw=blue, opacity=0.4] (3/4,1/2) -- (3/4,1) -- (1,1) -- (1,1/2) -- cycle;   
\draw[step=1,gray,very thin] (-5,-5) grid (1.5,1.5);   
\draw (0,1.7) node {$\ell_2$};   
\draw (1.7,0) node {$\ell_1$};   
\end{tikzpicture}
\par\end{centering}
\caption{The geometric representation of information structures in the space
of likelihood ratios\label{fig:The-geometric-representation}. The
left panel shows an example with 4 signals; the right panel shows
the incentives of the principal when the agent reveals $x$.}
\end{figure}
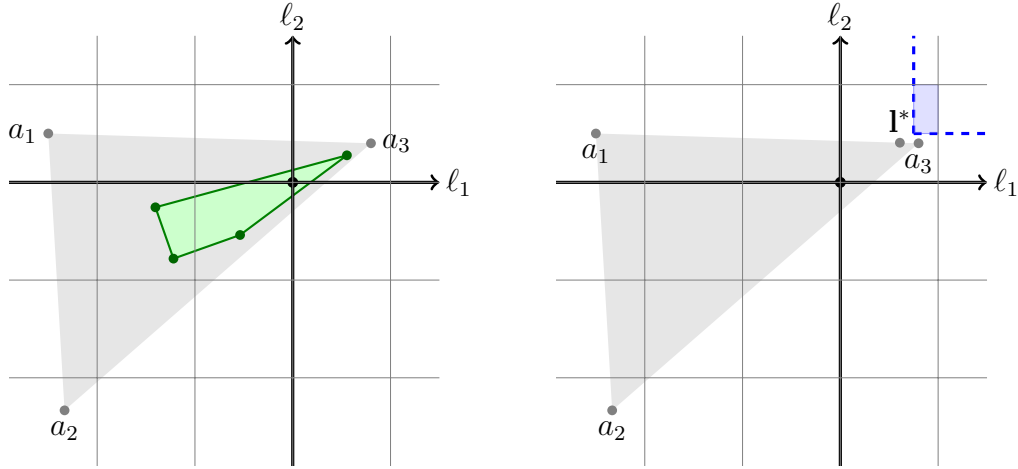
\end{example}

\subsection{Binary Information Structures}

We now exploit our geometric interpretation to show that we may restrict
the agent to choose information structures with at most two signals
without loss of generality. 
\begin{prop}
\label{prop:If--is}If $e^{*}$ is implementable by some information
structure $(S,\pi)$ and delivers expected wage $W\left(e^{*},\pi\right)$
to the agent, then $e^{*}$ is also implementable by a binary information
structure $\left(\hat{S},\hat{\pi}\right)$ with $|\hat{S}|=2$ and
$W\left(e^{*},\hat{\pi}\right)=W\left(e^{*},\pi\right)$.
\end{prop}
\begin{proof}
Suppose the effort level $e^{*}$ is implementable by $(S,\pi)$ and
consider the optimization problem in (\ref{wage_space}). Let $\mathbf{l}^{*}\in\text{co}\left(\mathbf{p}\right)$
be the optimal choice for the principal when choosing $e^{*}$. Given
the definition of $\text{co}\left(\mathbf{p}\right)$ and the geometric
game described above, there must exist $\left\{ \alpha_{s}\right\} _{s\in S}$
such that $\alpha_{s}\geq0$, $\sum_{s\in S}\alpha_{s}=1$ and
\[
\ell_{i}^{*}=1-\sum_{s\in S}\alpha_{s}\frac{p\left(s|e_{i}\right)}{p\left(s|e^{*}\right)},\forall i=1,\cdots,\left|E\right|.
\]

Let the information structure $\left(\hat{S},\hat{\pi}\right)$ be
defined as $\hat{S}=\left\{ L,H\right\} $ and 
\[
\hat{\pi}\left(H|x\right)=\sum_{s}\beta_{s}\pi\left(s|x\right),\hat{\pi}\left(L|x\right)=1-\hat{\pi}\left(H|x\right),
\]
where 
\[
\beta_{s}=\frac{\alpha_{s}/p\left(s|e^{*}\right)}{\sum_{s}\alpha_{s}/p\left(s|e^{*}\right)}.
\]
The probability function given this information structure satisfies
\begin{align*}
1-\frac{\hat{p}\left(H|e_{i}\right)}{\hat{p}\left(H|e^{*}\right)} & =1-\frac{\sum_{x}\hat{\pi}\left(H|x\right)f\left(x|e_{i}\right)}{\sum_{x}\hat{\pi}\left(H|x\right)f\left(x|e^{*}\right)},\\
 & =1-\frac{\sum_{s\in S}\frac{\alpha_{s}}{p\left(s|e^{*}\right)}\sum_{x}\pi\left(s|x\right)f\left(x|e_{i}\right)}{\sum_{s\in S}\frac{\alpha_{s}}{p\left(s|e^{*}\right)}\sum_{x}\pi\left(s|x\right)f\left(x|e^{*}\right)}\\
 & =1-\frac{\sum_{s\in S}\frac{\alpha_{s}}{p\left(s|e^{*}\right)}p\left(s|e_{i}\right)}{\sum_{s\in S}\frac{\alpha_{s}}{p\left(s|e^{*}\right)}p\left(s|e^{*}\right)}\\
 & =1-\sum_{s\in S}\frac{\alpha_{s}}{p\left(s|e^{*}\right)}p\left(s|e_{i}\right)=\ell_{i}^{*},\forall i=1,\cdots,\left|E\right|.
\end{align*}
The above implies that if the principal is to choose $e^{*}$, $\mathbf{l}^{*}$
is feasible under the new information structure $\left(\hat{S},\hat{\pi}\right)$,
i.e., $\mathbf{l}^{*}\in\text{co}\left(\hat{\mathbf{p}}\right)$.
Hence, in order to establish our claim, it is sufficient to show that
for any alternative $e_{i}\neq e^{*}$, $W\left(e_{i},\hat{\pi}\right)\geq W\left(e_{i},\pi\right)$.
To show this, it is sufficient to show that
\[
\left(1-\frac{\hat{p}\left(L|e_{1}\right)}{\hat{p}\left(L|e^{*}\right)},\cdots,1-\frac{\hat{p}\left(L|e_{\left|E\right|}\right)}{\hat{p}\left(L|e^{*}\right)}\right)\in\text{co}\left(\mathbf{p}\right)
\]
This would imply that $\text{co}\left(\hat{\mathbf{p}}\right)$ is
a subset of $\text{co}\left(\mathbf{p}\right)$ since $\text{co}\left(\hat{\mathbf{p}}\right)$
is the line that connects the above point to $\mathbf{l}^{*}$.

We have
\begin{align*}
1-\frac{\hat{p}\left(L|e_{i}\right)}{\hat{p}\left(L|e^{*}\right)} & =1-\frac{\sum_{x}\left(1-\hat{\pi}\left(H|x\right)\right)f\left(x|e_{i}\right)}{\sum_{x}\left(1-\hat{\pi}\left(H|x\right)\right)f\left(x|e^{*}\right)}\\
 & =1-\frac{\sum_{x}\sum_{s}\left(1-\beta_{s}\right)\pi\left(s|x\right)f\left(x|e_{i}\right)}{\sum_{x}\sum_{s}\left(1-\beta_{s}\right)\pi\left(s|x\right)f\left(x|e^{*}\right)}\\
 & =1-\frac{\sum_{s}\left(1-\beta_{s}\right)p\left(s|e_{i}\right)}{\sum_{s}\left(1-\beta_{s}\right)p\left(s|e^{*}\right)}\\
 & =1-\sum_{s}\frac{\left(1-\beta_{s}\right)p\left(s|e^{*}\right)}{\sum_{s}\left(1-\beta_{s}\right)p\left(s|e^{*}\right)}\frac{p\left(s|e_{i}\right)}{p\left(s|e^{*}\right)},
\end{align*}
which establishes the claim. This concludes the proof.
\end{proof}
We can describe the intuition behind the above proof graphically.
Consider an implementable effort $e^{*}$ and suppose that its associated
information structure is as depicted in Figure \ref{fig:The-intuition-for};
the green area represents the convex hull corresponding to this information
structure, $\text{co}\left(\mathbf{p}\right)$. Suppose that point
$b$ is the point of optimality for the principal in $\text{co}\left(\mathbf{p}\right)$
if he were to implement $e^{*}$. The red line represents the new
information structure $\hat{\pi}$ -- since there are only two signal
realizations, this has to be a line. If the principal is to implement
$e^{*}$, since $b\in\text{co}\left(\hat{\mathbf{p}}\right)$ and
$b$ is chosen under $\pi$, $b$ remains optimal. Moreover, since
$\text{co}\left(\hat{\mathbf{p}}\right)\subset\text{co}\left(\mathbf{p}\right)$,
for any other effort $e_{i}\neq e^{*}$, minimized cost under $\hat{\pi}$
must be at least as high as that under $\pi$. This, in turn, implies
that $e^{*}$ is implementable under $\hat{\pi}$ and implements the
same outcome as under $\pi$.

\begin{figure}

\begin{centering}
\begin{tikzpicture}[scale=1.5]
\clip (-2.9,-2.9) rectangle (2.,2.);
\fill[fill=gray!20!white] (0.4,0.22) -- (0.58,-1.5) -- (-7/3,-7/3) -- (-2.5,0.5) -- cycle;
\filldraw[fill=green!20!white, draw=green!50!black,thick] (-1.8,0) -- (-1.9,-2) -- (-.7,-1.5) -- (0.2, -0.5) -- cycle;
\fill[fill=gray] (0.58,-1.5) circle [radius=0.05] node[anchor=west]{$a_4$};
\fill[fill=gray] (0.4,.22) circle [radius=0.05] node[anchor=west]{$a_3$};
\fill[fill=gray] (-7/3,-7/3) circle [radius=0.05]node[anchor=north]{$a_2$};
\fill[fill=gray] (-2.5,.5) circle [radius=0.05] node[anchor=east]{$a_1$};
\fill[fill=green!40!black] (-1.8,0.) circle [radius=0.05];
\fill[fill=green!40!black] (-1.9,-2) circle [radius=0.05];
\fill[fill=green!40!black] (-0.7,-1.5) circle [radius=0.05];
\fill[fill=green!40!black] (0.2, -0.5) circle [radius=0.05];
\fill[fill=red!60!black,text=red!60!black] (-.3,-.375) circle [radius=0.05] node[above]{$b$};
\fill[fill=red!60!black] (-1.7,-1.625) circle [radius=0.05];
\draw[red!60!black,  thick] (-.3,-.375) -- (-1.7,-1.625);
\draw[->, very thick] (-10,-1.) -- (1.5,-1.);
\draw[->, very thick] (-1.,-10) -- (-1.,1.5);
\draw[step=1,gray,very thin] (-5,-5) grid (1.5,1.5);
\draw (-1,1.7) node {$\ell_2$};
\draw (1.7,-1) node {$\ell_1$};
\end{tikzpicture}\caption{The intuition for the construction of a two-point signal\label{fig:The-intuition-for}}
\par\end{centering}
\end{figure}
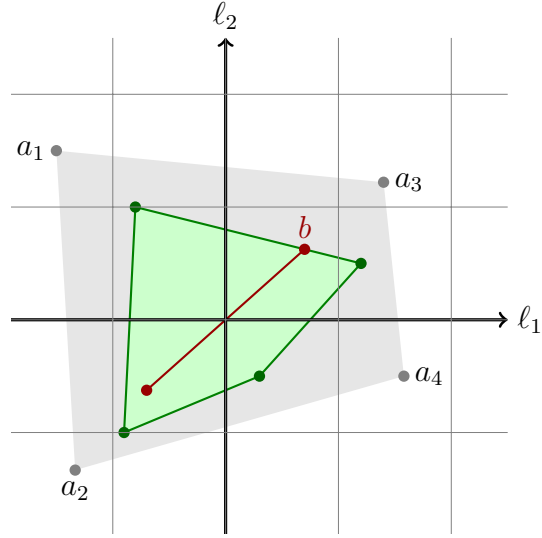

The above observations imply that if the principal prefers to implement
effort $e^{*}$ under the information structure $\pi$, he would prefer
the same under the binary information structure $\hat{\pi}$ and the
expected wage he has to pay to achieve this is the same under both
information structures. In what follows, we will use this result to
characterize the optimal information structures.

\textbf{Example 1 (Continued). }Recall Example \ref{exa:Suppose-that-,}
in which we argued that the principal -- under full revelation --
chooses to implement $e_{2}$ and thus his choice is inefficient (since
total surplus under $e_{2}$ is lower than that under $e_{3}$.)

It still remains to be seen whether $e_{3}$ is implementable and
how well the agent can do by choosing an information structure. By
Proposition \ref{prop:If--is}, we can focus on information structures
that only have two support points. Note that since the expected output
under $e_{1}$ is $\mathbb{E}\left[g(x)|e_{1}\right]=0.8$ and its
cost is 0, the principal can always guarantee $\mathbb{E}\left[g(x)|e_{1}\right]=0.8$. 

Now consider the point $\hat{\mathbf{l}}$ satisfying
\[
\frac{0.2}{\hat{\ell}_{2}}=\frac{0.3}{\hat{\ell}_{1}}=0.85
\]
At this point, the payoff to the principal of implementing each effort
level is given by
\begin{align*}
e_{3}:1.65-\max\left\{ \frac{0.2}{\hat{\ell}_{2}},\frac{0.3}{\hat{\ell}_{1}}\right\}  & =0.8\\
e_{2}:1.40-\left(1-\hat{\ell}_{2}\right)\cdot\frac{0.1}{\hat{\ell}_{1}-\hat{\ell}_{2}} & =1.40-0.65=0.75\\
e_{1}:0.8-0 & =0.8
\end{align*}
This implies that at $\hat{\mathbf{l}}$, $e_{3}$ maximizes the payoff
of the principal. Thus if $\hat{\mathbf{l}}\in\text{co}\left(\mathbf{f}\right)$,
we can choose an information structure that implements $e_{3}$. In
this example, this is indeed the case. Figure \ref{fig:Agent-optimal-information}
depicts the point $\hat{\mathbf{l}}$ as well as a two-point information
structure that implements it; the red line going through $\hat{\mathbf{l}}$.
If we set $S=\left\{ L,H\right\} $ and
\[
\pi\left(H|x\right)=\begin{cases}
\frac{27}{83}, & \text{if}\ x=x_{1}\\
\frac{15}{83}, & \text{if}\ x=x_{2}\\
\frac{41}{83}, & \text{if\ }x=x_{3}
\end{cases},\pi\left(L|x\right)=1-\pi\left(H|x\right),
\]
it can readily be checked that $\left(S,\pi\right)$ is associated
with the one depicted in Figure \ref{fig:Agent-optimal-information}.

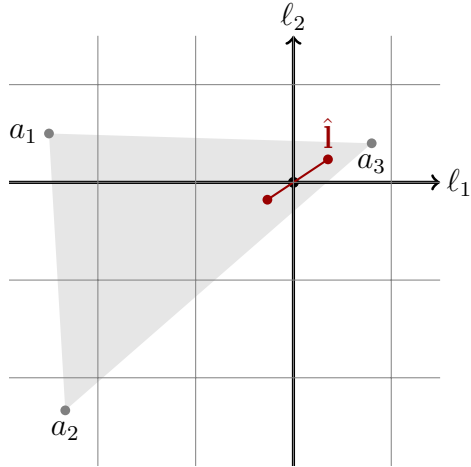
\begin{figure}[t]
\centering{}\begin{tikzpicture}[scale=1.3]
\clip (-2.9,-2.9) rectangle (2.7,2.7);   
\fill[fill=gray!20!white] (0.8,0.4) -- (-7/3,-7/3) -- (-2.5,0.5) -- cycle;   
\draw[->, very thick] (-10,0) -- (1.5,0);   
\draw[->, very thick] (0,-10) -- (0,1.5);   
\filldraw (0,0) circle [radius=0.05];
\fill[fill=red!60!black,text=red!60!black] (6/17,4/17) circle [radius=0.05] node[above]{$\hat{\mathbf{l}}$};
\fill[fill=red!60!black] (-0.2664,-0.1776) circle [radius=0.05];
\draw[red!60!black,thick] (6/17,4/17) -- (-0.2664,-0.1776);
\fill[fill=gray] (0.8,.4) circle [radius=0.05] node[below]{$a_3$};
\fill[fill=gray] (-7/3,-7/3) circle [radius=0.05]node[anchor=north]{$a_2$};
\fill[fill=gray] (-2.5,.5) circle [radius=0.05] node[anchor=east]{$a_1$};
\draw[step=1,gray,very thin] (-5,-5) grid (1.5,1.5);   
\draw (0,1.7) node {$\ell_2$};   
\draw (1.7,0) node {$\ell_1$};   
\end{tikzpicture}\caption{Agent-optimal information structure in Example 1 \label{fig:Agent-optimal-information}}
\end{figure}

The above example illustrates the agent's power the choice of the
information structure provides. By choosing the aforementioned information
structure, not only the agent is able to implement $e_{3}$, i.e.,
the efficient level of effort, but also she is able to capture the
entire surplus. In what follows, we show that under some conditions
on the performance technology, this is always possible.

\section{Efficient Surplus Extraction\label{sec:Efficient-and-Full}}

In this section, we provide sufficient conditions on the performance
technology $f\left(\cdot|\cdot\right)$ and cost function $c\left(\cdot\right)$
so that the agent is able to implement the first-best effort and extract
the entire surplus. 

Let $e^{*}$ be the first-best level of effort that satisfies
\begin{align*}
e^{*}\in\arg\max_{e\in E}\quad & \mathbb{E}\left[g(x)|e\right]-c\left(e\right).
\end{align*}
Suppose that the agent wishes to implement $e^{*}$ and capture the
entire surplus given by $\mathbb{E}\left[g(x)|e^{*}\right]-c\left(e^{*}\right)-\mathbb{E}\left[g(x)|e_{1}\right]$.
For this to occur, we need to choose $\text{co}\left(\mathbf{p}\right)$
and $\mathbf{l}^{*}\in\text{co}\left(\mathbf{p}\right)$ that satisfy
\begin{align*}
\mathbb{E}\left[g(x)|e_{1}\right] & =\mathbb{E}\left[g\left(x\right)|e^{*}\right]-\max_{i:\ell_{i}^{*}\geq0}\frac{c\left(e^{*}\right)-c\left(e_{i}\right)}{\ell_{i}^{*}}\\
 & \geq\max_{e_{j}\in E,\mathbf{l}\in\text{co}\left(\mathbf{p}\right)\cap\Omega_{j}}\mathbb{E}\left[g\left(x\right)|e_{j}\right]-\left(1-\ell_{j}\right)\max_{i:\ell_{i}>\ell_{j}}\frac{c\left(e_{j}\right)-c\left(e_{i}\right)}{\ell_{i}-\ell_{j}}
\end{align*}

Consider the likelihood vector $\mathbf{l}^{*}$ that satisfies the
following property:
\begin{align*}
\mathbf{\mathbb{E}}\left[g\left(x\right)|e^{*}\right]-\mathbb{E}\left[g(x)|e_{1}\right] & =\frac{c\left(e^{*}\right)-c\left(e_{i}\right)}{\ell_{i}^{*}},\forall i=1,\cdots,\left|E\right|.
\end{align*}
In words, this is a point in which the agent is indifferent among
all the efforts -- see the reformulation of (\ref{ICA}) in (\ref{eq:ICAr}).
Moreover, the principal's payoff is his outside option and thus if
the principal chooses to implement $e^{*}$ and chooses $\mathbf{l}^{*}\in\text{co}\left(\mathbf{p}\right)$,
the agent captures the entire surplus. In the following proposition,
we show that if $\mathbf{l}^{*}\in\text{co}\left(\mathbf{f}\right)$,
we can construct an information structure -- and its associated $\text{co}\left(\mathbf{p}\right)$
-- in which $\mathbf{l}^{*}$ is the best choice for the principal
and $e^{*}$ is the best level of effort.
\begin{thm}
\label{thm:Suppose-that-.}Suppose that $\mathbf{l}^{*}\in\text{co}\left(\mathbf{f}\right)$.
Then $e^{*}$ is implementable and there exists an information structure
$\left(S,\pi\right)$ for which the agent's payoff is $U_{A}=\mathbb{E}\left[g\left(x\right)|e^{*}\right]-\mathbb{E}\left[g\left(x\right)|e_{1}\right]-c\left(e^{*}\right)$
and the principal's payoff is $U_{P}=\mathbb{E}\left[g\left(x\right)|e_{1}\right]$,
i.e., the agent can capture the entire surplus.
\end{thm}
\begin{proof}
First, we note that by Lemma \ref{lem:Let--represnt} in the Appendix,
$\mathbf{0}\in\text{co}\left(\mathbf{f}\right)$ is an interior point
of this convex set, where interiority is defined in an appropriatley
defined subspace of $\mathbb{R}^{m}$. We define the following convex
hull (and by Proposition \ref{prop:For-any-information}, its associated
information structure): 
\[
\text{co}\left(\mathbf{p}\right)=\left\{ \lambda\mathbf{l}^{*}+\left(1-\lambda\right)\left(-\alpha\mathbf{l}^{*}\right):\forall\lambda\in\left[0,1\right]\right\} ,
\]
where $\alpha>0$ is such that $-\alpha\mathbf{l}^{*}\in\text{co}\left(\mathbf{f}\right)$.
Such an $\alpha$ always exists due to $\mathbf{0}$ being an interior
point. Since $\text{co}\left(\mathbf{p}\right)$ is a line through
$\mathbf{l}^{*}$ and the origin, all of its members must satisfy
\begin{equation}
\forall i,j,\quad\frac{\ell_{i}}{\ell_{j}}=\frac{c\left(e^{*}\right)-c\left(e_{i}\right)}{c\left(e^{*}\right)-c\left(e_{j}\right)}=\frac{\ell_{i}^{*}}{\ell_{j}^{*}}\Rightarrow\frac{\ell_{i}}{\ell_{1}}=\frac{c\left(e^{*}\right)-c\left(e_{i}\right)}{c\left(e^{*}\right)},\forall i.\label{eq: line}
\end{equation}
This implies that the cost of choosing any point $\mathbf{l}\in\text{co}\left(\mathbf{p}\right)$
to implement any effort $e_{j}\in E$ is given by
\begin{align*}
\left(1-\ell_{j}\right)\max_{i:\ell_{i}\geq\ell_{j}}\frac{c\left(e_{j}\right)-c\left(e_{i}\right)}{\ell_{i}-\ell_{j}} & =\left(1-\frac{c\left(e^{*}\right)-c\left(e_{j}\right)}{c\left(e^{*}\right)}\ell_{1}\right)\frac{c\left(e^{*}\right)}{\ell_{1}}\\
 & =\frac{c\left(e^{*}\right)}{\ell_{1}}+c\left(e_{j}\right)-c\left(e^{*}\right).
\end{align*}

Since choice of $\mathbf{l}$ under $e_{j}$ must be a member of the
cone $\Omega_{j}$, we must have $\ell_{j}\leq\ell_{1}$ and this
combined with (\ref{eq: line}) implies that $\ell_{1}\geq0$. Thus
the above expression is minimized at $\mathbf{l}^{*}$. Hence, the
highest payoff of the principal from choosing $e_{j}$ is given by
\begin{align*}
\mathbb{E}\left[g\left(x\right)|e_{j}\right]-c\left(e_{j}\right)-\frac{c\left(e^{*}\right)}{\ell_{1}^{*}}+c\left(e^{*}\right) & =\\
\mathbb{E}\left[g\left(x\right)|e_{j}\right]-c\left(e_{j}\right)+c\left(e^{*}\right)-\mathbb{E}\left[g\left(x\right)|e^{*}\right]+\mathbb{E}\left[g(x)|e_{1}\right]
\end{align*}
The above is maximized at $e_{j}=e^{*}$, which implies that $e^{*}$
can be implemented with the chosen information structure. Moreover,
the payoff of the principal is given by $\mathbb{E}\left[g\left(x\right)|e_{1}\right]$.
This concludes the proof.
\end{proof}
The proof of the above theorem emphasizes the power of the agent when
$\text{\ensuremath{\mathbf{l}^{*}\in\text{co}\left(\mathbf{f}\right)}}$.
By being able to control the information structure, the agent can
control the wage that is needed for the principal to implement his
desired effort. By choosing $\mathbf{l}^{*}$, the agent is forcing
the principal to fully compensate the agent for her cost of effort.
This implies that the payoff of the principal becomes total surplus
shifted by a constant. Hence, it is optimal to choose the first best
level of effort.

While Theorem \ref{thm:Suppose-that-.} provides sufficient conditions
for implementability of the first-best effort and full surplus extraction,
it is not immediately evident what it imposes on the structure of
the model. In what follows, we try to shed light on this.

\paragraph*{Almost Perfect Performance Technology.}

Suppose that the performance technology satisfies the following property
\begin{align*}
X & =E\subset\mathbb{R}_{+},e_{1}=0,\\
f\left(e_{j}|e_{j}\right) & =1-\left(m-1\right)\varepsilon,\forall j=1,\cdots,m,\\
f\left(e_{i}|e_{j}\right) & =\varepsilon,\forall i\neq j,
\end{align*}
where $1/\left(m-1\right)>\varepsilon>0$. In words, the above performance
technology puts probability $\varepsilon$ on $e_{i}\neq e_{j}$ if
$e_{j}$ is chosen. As $\varepsilon$ converges to 0, it converges
to a setting where observing $x\in X$ fully reveals $e$, i.e., \emph{a
perfect performance technology}. Note that for $\varepsilon$ small
enough, the first-best level of effort is given by
\[
e_{j}\in\arg\max_{e\in E}e-c\left(e\right)
\]
Suppose that in the above $e_{j}>e_{1}=0$. Moreover, as $\varepsilon$
converges to 0, the point $\mathbf{l}^{*}$ converges to
\[
\ell_{i}^{*}=\frac{c\left(e_{j}\right)-c\left(e_{i}\right)}{e_{j}}
\]
Finally, note that for any $\varepsilon$, the set $\text{co}\left(\mathbf{f}\right)$
is the convex hull of the following likelihood ratios:
\[
1-\frac{f\left(e_{k}|e_{i}\right)}{f\left(e_{k}|e_{j}\right)}=\begin{cases}
1-\frac{\varepsilon}{1-\left(m-1\right)\varepsilon}, & \text{if}\ k=j,k\neq i\\
0, & \text{if\ }k=i=j\\
0, & \text{if\ }k\neq j,k\neq i\\
m-\frac{1}{\varepsilon}, & \text{if\ }k\neq j,k=i
\end{cases}
\]
As $\varepsilon$ converges to 0, the above set gets larger and the
point associated with $k=j$ converges to 
\[
\left(\underbrace{1,\cdots,1}_{j-1\text{ times}},0,1,\cdots,1\right),
\]
 while the points associated with $k\neq j$ converge to 
\[
\left(\underbrace{0,\cdots,0}_{k-1\text{ times}},-\infty,0,\cdots,0\right).
\]
This implies that as $\varepsilon$ converges to 0, $\text{co}\left(\mathbf{f}\right)\rightarrow\left(-\infty,1\right]^{j-1}\times\left\{ 0\right\} \times\left(-\infty,1\right]^{m-j}$.
This is depicted in Figure \ref{fig:Almost-Perfect-Performance}.
Since $\ell_{i}^{*}\leq\frac{c\left(e_{j}\right)}{e_{j}}<1$, for
$\varepsilon$ small enough, we must have that $\mathbf{l}^{*}\in\text{co}\left(\mathbf{f}\right)$.
We thus have the following corollary to Theorem \ref{thm:Suppose-that-.}:
\begin{cor}
\label{cor:Let--be}Let $f\left(\cdot|\cdot\right)$ be an almost
perfect performance technology satisfying 
\begin{align*}
X & =E\subset\mathbb{R}_{+},e_{1}=0,\\
f\left(e_{j}|e_{j}\right) & \leq1-\left(m-1\right)\varepsilon,\forall j=1,\cdots,m,\\
f\left(e_{i}|e_{j}\right) & \geq\varepsilon,\forall i\neq j.
\end{align*}
Then there exists $\overline{\varepsilon}$ such that for all $\varepsilon\leq\overline{\varepsilon}$,
the agent can implement the first-best level of effort and capture
the entire surplus.

\begin{figure}
\begin{centering}
\begin{tikzpicture}[scale=1.3]
\fill[pattern color=orange,  pattern=grid] (-3,1.) -- (1.,1.) --(1.,-3) -- (-3,-3) -- cycle;
\clip (-2.9,-2.9) rectangle (2.7,2.7);   
\fill[fill=gray!20!white] (-7,0) -- (0,-7) -- (7/8,7/8) -- cycle;   
\fill[fill=green!20!white] (-2,0) -- (0,-2) -- (2/3,2/3) -- cycle;
\draw[->, very thick] (-10,0) -- (1.5,0);   
\draw[->, very thick] (0,-10) -- (0,1.5);   
\filldraw (0,0) circle [radius=0.05];
\filldraw[fill=red!60!black,text=red!60!black] (0.8,0.7) circle [radius=0.05] node[below]{$\mathbf{l}^*$};
\draw[step=1,gray,very thin] (-5,-5) grid (1.5,1.5);   
\draw (0,1.7) node {$\ell_2$};   
\draw (1.7,0) node {$\ell_1$};  
\draw[thin] (-5,1.) -- (1.,1.) --(1.,-5);
\end{tikzpicture}
\par\end{centering}
\caption{\label{fig:Almost-Perfect-Performance}Almost Perfect Performance
Technology for $m=3$. The grey and green shaded areas represent $\text{co}\left(\mathbf{f}\right)$
as $\varepsilon$ becomes smaller. As $\varepsilon\rightarrow0$,
$\text{co}\left(\mathbf{f}\right)$ converges to the shaded quarter-space
south-west of $\left(1,1\right)$. }
\end{figure}
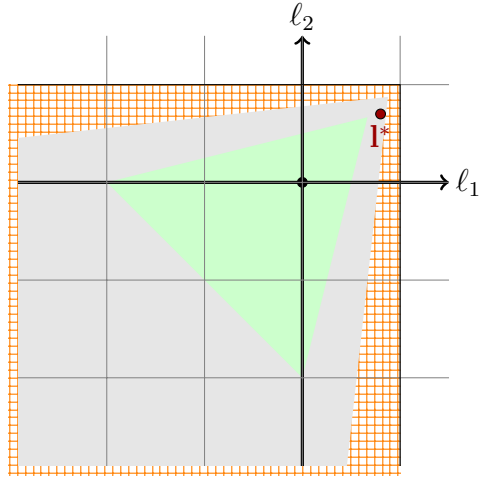
\end{cor}

\section{Continuous Effort and Output\label{sec:Continuous-Effort-and-1}}

In this section, we consider a version of the model from Section \ref{sec:Model}
where the effort space and the output space are continuous. Applying
a  first-order approach, we derive sufficient conditions such that
the optimal indicator structure takes the form of monotone or hump-shaped
threshold signals. In other words, we derive conditions such that
the optimal information structure is characterized by an indicator
with at most two thresholds.

Consider a principal employing an agent to perform a task whose output
is represented by a real number $x\in X=[0,1]$. The agent chooses
effort $e\in E=[0,1]$ to perform the task. The agent's effort choice
induces a probability distribution $f(x|e)$ over the output space
$X$, whose density function $f(x|e)$ is assumed to be twice differentiable
with respect to $e$ for every $x\in X$. Effort is costly to the
agent; the cost of effort $e$ is given by $c(e)$ for some real-valued
function $c:E\rightarrow\mathbb{R}_{+}$ where $c'(e)\geq0,\ c''(e)>0,\ \forall e\in E$.
For simplicity, let $c(0)=0$. The timing and payoff functions of
the game are otherwise the same as in the model in Section \ref{sec:Model}.
To proceed, we assume that the first-order approach is valid for all
the optimization problems faced by the agent and the principal.

We characterize this continuous version of the game in the same manner
as the discrete case. Given the result in Proposition \ref{prop:If--is},
one can use an approximation argument to show that optimal signal
is binary, i.e., $S=\{L,H\}$. As a result, let $p\left(e\right):E\rightarrow\left[0,1\right]$
represent the probability of $s=H$ induced by a given information
structure $(S,\pi)$ and the underlying probability distribution of
outcomes given effort $f(x|e)$. Note, for any level of effort $e$,
$p\left(e\right)=\int_{0}^{1}f\left(x|e\right)\pi\left(H|x\right)dx$.
Differentiating the stochastic mapping $p$ with respect to $e$ yields
\[
p'\left(e\right)=\int_{0}^{1}f_{e}\left(x|e\right)\pi\left(H|x\right)dx,\ \forall e\in E.
\]

The agent's problem in the last stage of the game, given an information
structure and a compensation scheme, is
\[
\max_{e\in E}\ p\left(e\right)w\left(H\right)+\left[1-p\left(e\right)\right]w\left(L\right)-c\left(e\right)
\]
The first-order condition characterizing the agent's optimal effort
is
\[
p'\left(e\right)\left[w\left(H\right)-w\left(L\right)\right]=c'(e).
\]

For any desired level of effort, the principal chooses a compensation
scheme that solves
\begin{align}
\begin{split}\min_{w}\  & p\left(e\right)w\left(H\right)+\left(1-p\left(e\right)\right)w\left(L\right)\\
\text{s.t}\  & p'\left(e\right)\left[w\left(H\right)-w\left(L\right)\right]=c'\left(e\right),\\
 & w\left(H\right),w\left(L\right)\geq0.
\end{split}
\label{eq: Principal Problem}
\end{align}

If $\lambda(e,\pi)$ denotes the Lagrange multiplier on the agent's
incentive compatibility constraint in (\ref{eq: Principal Problem}),
then 
\[
\lambda(e,\pi)=\min\left\{ \frac{p\left(e\right)}{p'\left(e\right)},-\frac{1-p\left(e\right)}{p'\left(e\right)}\right\} .
\]
 Exactly as in the discrete case, the above reveals that the agent's
choice of information structure, which induces a set of likelihood
ratios dependent on the stochastic mapping $p$, determines the principal's
shadow cost of incentivizing a given effort level.

If we define the expected payoff to the principal from a given effort
level $e$ as $\mathbb{E}\left[g(x)|e\right]=g(x)f(x|e)dx$, then
the principal's optimal choice of effort solves
\[
\max_{e}\mathbb{E}\left[g\left(x\right)|e\right]-\lambda\left(e,\pi\right)c'\left(e\right)
\]
with associated optimality condition

\[
\frac{\partial\mathbb{E}\left[g(x)|e\right]}{\partial e}-\frac{\partial\lambda(e,\pi)}{\partial e}c'(e)-\lambda(e,\pi)c''(e)=0.
\]

Using the above sequence of optimality conditions, we obtain the following
optimization problem that describes the agent's choice of effort and
information structure in the first stage of the game:
\begin{align}
\begin{split}\max_{e,\pi}\  & \lambda(e,\pi)c'(e)-c(e)\ \text{s.t.}\\
 & \frac{\partial\mathbb{E}[g(x)|e]}{\partial e}-\frac{\partial\lambda(e,\pi)}{\partial e}c'(e)-\lambda(e,\pi)c''(e)=0,\\
 & \frac{1}{\lambda(e,\pi)}=\max\left\{ \frac{\int_{0}^{1}f(x|e)\pi(x)dx}{\int_{0}^{1}f_{e}(x|e)\pi(x)dx},-\frac{1-\int_{0}^{1}f(x|e)\pi(x)dx}{\int_{0}^{1}f_{e}(x|e)\pi(x)dx}\right\} .
\end{split}
\label{eq:Agent Problem lambda}
\end{align}
The following assumption allows us to provide a sharp characterization
of the optimal information structure:
\begin{assumption}
\label{assu:Given-any-effort}Given any effort $e\in E$, the likelihood
$\frac{f_{e}(x|e)}{f(x|e)}$ is strictly monotone in output $x$ and
its derivative $\frac{\partial}{\partial e}\frac{f_{e}(x|e)}{f(x|e)}$
is a convex function of the likelihood $\frac{f_{e}(x|e)}{f(x|e)}$.
\end{assumption}
Several distribution functions satisfy Assumption \ref{assu:Given-any-effort}.
Examples include power distributions: $f\left(x|e\right)=ex^{e-1}$,
and truncated exponential distributions: $f\left(x|e\right)=\frac{d}{dx}\frac{e^{x}-1}{e-1}$.

Our main characterization of the optimal information structure is
as follows:
\begin{prop}
Suppose that Assumption \ref{assu:Given-any-effort} holds. Then,
the equilibrium information structure is characterized by at most
two thresholds in the output space. If the equilibrium information
structure has a single threshold, say $x^{*}$, then $\pi\left(H|x\right)=1$
if and only if $x\geq x^{*}$. If the equilibrium information structure
has two thresholds, say $(x_{1}^{*},x_{2}^{*})$ then $\pi\left(H|x\right)=1$
if and only if $x\in[x_{1}^{*},x_{2}^{*}]$.
\end{prop}
\begin{proof}
See the Appendix.
\end{proof}
Given an effort choice, the agent's optimal choice of information
structure requires the agent to choose the probability of being paid
at any output such that there is no net marginal benefit from a marginal
change in this probability. Generally, it is not possible to satisfy
this condition for all the output levels. The agent ends up choosing
the extreme probabilities for every output, depending on whether her
net marginal benefit is increasing or decreasing in the probability
of being paid at that output level. Assumption 2 ensure that the output
intervals at which the agent uses either of the extreme probabilities
are structured such that the optimal information structure takes the
form of monotone or hump-shaped thresholds signals.
\begin{example}
Let $f(x|e)=3ex^{3e-1}$. As mentioned earlier, this distribution
satisfies Assumption \ref{assu:Given-any-effort}. If the effort cost
is $c(e)=\frac{e^{2}}{2}$, the equilibrium information structure
is characterized by $x^{*}=0.45$ where $\pi\left(H|x\right)=1$ if
and only if $x\geq x^{*}$. The principal pays the agent $w=0.1048$
if he receives the high signal; otherwise, he pays nothing. This choice
implements effort $e^{*}=0.2725$ yielding the following payoffs:
$U_{P}=0.3450$, $U_{A}=0.0677$. For comparison, the first-best effort
is $e=0.4708$. In this case, and opposite to that of section \ref{sec:Efficient-and-Full},
the agent is unable to implement the efficient outcome. 
\end{example}

\section{Conclusion\label{sec:Conclusion}}

In this paper, we have developed the theoretical tool in the design
of contracts in principal-agent settings. Part of the contract design
is what indicators should be used as contingencies for payments. Despite
the importance of this question and its relevance for analysis of
contracting decisions, this part of the contracting procedure is less
explored. 

Our paper has two broad implications. First, our methodology of thinking
about likelihood ratios can be applied to other settings in which
considerations of communication off the equilibrium path are important.
In our setup, unlike other models of communication and information
design, the design of information structure for the equilibrium level
of effort affects off-path communication and the ability of the principal
and the agent to capture surplus. This can arise in other settings
with strategic information transmission and our method can be useful
for that.

Second, our paper ties the choice of indicators in contracting to
the bargaining power of the parties. In the textbook moral hazard
problem, principal makes a take-it-or-leave it offer. As a result,
a version of the informativeness principle often holds; the principal
wishes to use all the information available, if possible, as contingency
for payments. In contrast, in our model, the agent has different incentives
for choice of indicators. There are some casual observations that
are in line with this explanation. For example, contracts in NFL are
often extremely detailed and payments to football players are highly
contingent on various measures of individual and team outcomes. In
contrast, contracts found in the English Premier League are not as
detailed. They are often contingent on very coarse personal outcomes
such as the number of goals scored reaching a particular threshold.
In light of our theory, the level of competition in English/European
soccer (in the form of increased player bargaining power) compared
to a lack thereof in NFL could be behind this observation. Future
work can hopefully shed light on the importance of this channel in
the data.

\newpage{}

\bibliographystyle{ecta}
\bibliography{moralhazardrefs}

\section{Appendix}
\begin{lem}
\label{lem:Let--represnt}Let $T$ represent the lowest-dimensional
linear subspace in $\mathbb{\mathbb{R}}^{|E|}$ that contains $\text{co}\left(\mathbf{f}\right)$.
If $f(\cdot|e^{*})$ is full-support, then the origin is an interior
point of $co\left(\mathbf{f}\right)$ with respect to $T$.
\end{lem}
\begin{proof}
Notice that the origin can be written as a convex combination of the
points defining $\text{co}\left(\mathbf{f}\right)$ with weights $f(x|e^{*})$:

\begin{align*}
\sum_{x}f(x|e^{*})\bigg(1-\frac{f(x|e_{i})}{f(x|e^{*})}\bigg) & =\sum_{x}f(x|e^{*})-\sum_{x}f(x|e_{i})=1-1=0,\ \forall i,
\end{align*}
where $f(x|e^{*})>0,\ \forall x\in X$ by the full-support assumption.
Therefore, the origin is always included in the convex set $\text{co}\left(\mathbf{f}\right)$.
Suppose by contradiction that the origin is not an interior point
of $\text{co}\left(\mathbf{f}\right)$ with respect to $T$. By the
supporting hyperplane theorem, there exists a hyperplane in the linear
subspace $T$ that contains the origin and $\text{co}\left(\mathbf{f}\right)$
is entirely contained in one of the two closed half-spaces bounded
by the hyperplane. However, since $f(x|e^{*})>0,\ \forall x\in X$,
this is possible only if $\text{co}\left(\mathbf{f}\right)$ is entirely
contained in this hyperplane. This is a contradiction because $T$
is the lowest-dimensional linear subspace in $\mathbb{\mathbb{R}}^{|E|}$
that contains $co\left(\mathbf{f}\right)$.
\end{proof}
\textbf{Proof of Proposition 3} : Let $\pi(x)=\pi(H|x),\ \forall x\in X$.
Without loss of generality, for any desired implementable effort level
$e$, we impose  $\int_{X}f_{e}(x|e)\pi(x)dx\geq0$. Consequently,
we have 
\[
\lambda(e,\pi)=\frac{\int_{X}f(x|e)\pi(x)dx}{\int_{X}f_{e}(x|e)\pi(x)dx}\geq0.
\]
We may then write the agent's problem in (\ref{eq:Agent Problem lambda})
as 
\begin{align}
\begin{split}\max_{e,\pi}\  & \lambda(e,\pi)c'(e)-c(e)\ \text{s.t.}\\
 & \frac{\partial E[g(x)|e]}{\partial e}-\frac{\partial\lambda(e,\pi)}{\partial e}c'(e)-\lambda(e,\pi)c''(e)=0,\\
 & \lambda(e,\pi)=\frac{\int_{X}f(x|e)\pi(x)dx}{\int_{X}f_{e}(x|e)\pi(x)dx},\\
 & \lambda(e,\pi)\geq0,\ 0\leq\pi(x)\leq1,\ \forall x\in X.
\end{split}
\label{eq: Agent Problem Binary}
\end{align}
We write the Lagrangian corresponding to (\ref{eq: Agent Problem Binary}),
momentarily ignoring the inequality constraints: 
\[
\mathcal{L}(e,\pi,\eta)=\lambda(e,\pi)c'(e)-c(e)+\eta\bigg[\frac{\partial\mathbb{E}[g(x)|e]}{\partial e}-\frac{\partial\lambda(e,\pi)}{\partial e}c'(e)-\lambda(e,\pi)c''(e)\bigg].
\]
For every output $\tilde{x}\in X$, the agent's optimal information
structure in (\ref{eq: Agent Problem Binary}) must satisfy

\begin{align}
\frac{\partial\mathcal{L}(e,\pi,\eta)}{\partial\pi(\tilde{x})}=0, & \ \text{if}\ 0<\pi(\tilde{x})<1,\nonumber \\
\frac{\partial\mathcal{L}(e,\pi,\eta)}{\partial\pi(\tilde{x})}\geq0, & \ \text{if}\ \pi(\tilde{x})=1,\label{eq: FOC pi}\\
\frac{\partial\mathcal{L}(e,\pi,\eta)}{\partial\pi(\tilde{x})}\leq0, & \ \text{if}\ \pi(\tilde{x})=0.\nonumber 
\end{align}
With some algebra, we get
\begin{proof}
\begin{equation}
\frac{\partial\mathcal{L}(e,\pi,\eta)}{\partial\pi(\tilde{x})}=f(\tilde{x}|e)\bigg[A_{1}(e,\pi,\eta)-A_{2}(e,\pi,\eta)\frac{f_{e}(\tilde{x}|e)}{f(\tilde{x}|e)}+A_{3}(e,\pi,\eta)\frac{f_{ee}(\tilde{x}|e)}{f(\tilde{x}|e)}\bigg]\label{eq: FOC Simple}
\end{equation}
where $A_{1},A_{2},A_{3}$ are some functions independent of $\tilde{x}$.
Note that since 
\[
\frac{f_{ee}(x|e)}{f(x|e)}=\frac{\partial}{\partial e}\frac{f_{e}(x|e)}{f(x|e)}+\bigg(\frac{f_{e}(x|e)}{f(x|e)}\bigg)^{2},
\]
we can write (\ref{eq: FOC Simple}) as 
\begin{align}
\frac{\partial\mathcal{L}(e,\pi,\eta)}{\partial\pi(\tilde{x})} & =f(\tilde{x}|e)\bigg[A_{1}(e,\pi,\eta)-A_{2}(e,\pi,\eta)\frac{f_{e}(\tilde{x}|e)}{f(\tilde{x}|e)}+\nonumber \\
 & \quad\quad\quad\quad\quad A_{3}(e,\pi,\eta)\bigg(\frac{f_{e}(\tilde{x}|e)}{f(\tilde{x}|e)}\bigg)^{2}+A_{3}(e,\pi,\eta)\frac{\partial}{\partial e}\frac{f_{e}(\tilde{x}|e)}{f(\tilde{x}|e)}\bigg].\label{eq: Prop 3}
\end{align}
Since $\frac{f_{e}(x|e)}{f(x|e)}$ is monotone, the function in (\ref{eq: Prop 3})
inherits the curvature of $\frac{\partial}{\partial e}\frac{f_{e}(x|e)}{f(x|e)}$.
That is, $\frac{1}{f(x|e)}\frac{\partial\mathcal{L}(e,\pi,\eta)}{\partial\pi(x)}$
is a convex or concave function of $\frac{f_{e}(x|e)}{f(x|e)}$ depending
on the sign of $A_{3}(e,\pi,\eta)$. As a result,  the sign of $\frac{\partial\mathcal{L}(e,\pi,\eta)}{\partial\pi(x)}$
changes at most twice over the interval $X$.

Note that if $\frac{\partial\mathcal{L}(e,\pi,\eta)}{\partial\pi(x)}$
is always positive or always negative (for a given effort level $e$),
then the information structure is fully uninformative. Such information
structures cannot be incentive compatible as they provide no incentives
for the agent to conduct any effort level $e>e_{1}$.  Consequently,
the equilibrium information structure has either one or two thresholds.
In either case, it follows immediately from (\ref{eq: FOC pi}) that
if the equilibrium information structure has a single threshold, say
$x^{*}$, then $\pi(x)=1$ if and only if $x\geq x^{*}$. If the equilibrium
information structure has two thresholds, say $(x_{1}^{*},x_{2}^{*})$,
then $\pi(x)=1$ if and only if $x\in[x_{1}^{*},x_{2}^{*}]$.
\end{proof}

\end{document}